\documentclass[runningheads,envcountsame]{llncs}

\usepackage[T1]{fontenc}
\usepackage{lmodern}

\usepackage{microtype}
\usepackage{transparent}
\renewcommand{\phi}{\varphi}
\renewcommand{\epsilon}{\varepsilon}

\usepackage{lineno}

\usepackage{amsmath}
\usepackage{amssymb}
\usepackage{mathtools}
\usepackage{ stmaryrd }
\usepackage{enumitem}

\usepackage{float}
\usepackage{booktabs}
\usepackage{multirow}
\usepackage{array}
\usepackage{fancybox}
\usepackage[skip=3pt]{caption}
\usepackage{subcaption}
\usepackage{tikz}
\usepackage{tkz-euclide}
\usepackage{tikz-3dplot}
\usepackage{tikz-qtree}
\usepackage{pgfplots}
\usepackage{makecell}
\usepackage[clock]{ifsym}
\usepackage{tabularx}
\usepackage{ltablex}
\usepackage[markup=addedmarkup=colored,deletedmarkup=sout,final]{changes}
\setaddedmarkup{{\bf \color{red}#1}}
\setdeletedmarkup{{\color{black!30!white}\sout{#1}}}
\usepackage{mdframed}
\usepackage{etoolbox}

\usepackage{rotating}
\tikzset{
    itria/.style={
    draw,dashed,shape border uses incircle,
    isosceles triangle,shape border rotate=90,yshift=-1.25cm},
    rtria/.style={
    draw,dashed,shape border uses incircle,
    isosceles triangle,isosceles triangle apex angle=80,
    shape border rotate=-40,yshift=0.1cm,xshift=0.35cm},
    ltria/.style={
    draw,dashed,shape border uses incircle,
    isosceles triangle,isosceles triangle apex angle=80,
    shape border rotate=-140,yshift=0.1cm,xshift=-0.35cm},
    ritria/.style={
    draw,dashed,shape border uses incircle,
    isosceles triangle,isosceles triangle apex angle=110,
    shape border rotate=-55,yshift=0.1cm},
    letria/.style={
    draw,dashed,shape border uses incircle,
    isosceles triangle,isosceles triangle apex angle=110,
    shape border rotate=235,yshift=0.1cm}
}

\usepackage{setspace}
\usepackage[ruled,noline,noend,linesnumbered,noresetcount]{algorithm2e}
\RestyleAlgo{algoruled}
\SetKwComment{Comment}{// }{}
\SetAlgoSkip{smallskip}
\SetAlFnt{\footnotesize}

\usepackage{xcolor} 
\definecolor{rwth-blue}{cmyk}{1,.5,0,0}\colorlet{rwth-lblue}{rwth-blue!50}\colorlet{rwth-llblue}{rwth-blue!25}
\definecolor{rwth-violet}{cmyk}{.6,.6,0,0}\colorlet{rwth-lviolet}{rwth-violet!50}\colorlet{rwth-llviolet}{rwth-violet!25}
\definecolor{rwth-purple}{cmyk}{.7,1,.35,.15}\colorlet{rwth-lpurple}{rwth-purple!50}\colorlet{rwth-llpurple}{rwth-purple!25}
\definecolor{rwth-carmine}{cmyk}{.25,1,.7,.2}\colorlet{rwth-lcarmine}{rwth-carmine!50}\colorlet{rwth-llcarmine}{rwth-carmine!25}
\definecolor{rwth-red}{cmyk}{.15,1,1,0}\colorlet{rwth-lred}{rwth-red!50}\colorlet{rwth-llred}{rwth-red!25}
\definecolor{rwth-magenta}{cmyk}{0,1,.25,0}\colorlet{rwth-lmagenta}{rwth-magenta!50}\colorlet{rwth-llmagenta}{rwth-magenta!25}
\definecolor{rwth-orange}{cmyk}{0,.4,1,0}\colorlet{rwth-lorange}{rwth-orange!50}\colorlet{rwth-llorange}{rwth-orange!25}
\definecolor{rwth-yellow}{cmyk}{0,0,1,0}\colorlet{rwth-lyellow}{rwth-yellow!50}\colorlet{rwth-llyellow}{rwth-yellow!25}
\definecolor{rwth-grass}{cmyk}{.35,0,1,0}\colorlet{rwth-lgrass}{rwth-grass!50}\colorlet{rwth-llgrass}{rwth-grass!25}
\definecolor{rwth-green}{cmyk}{.7,0,1,0}\colorlet{rwth-lgreen}{rwth-green!50}\colorlet{rwth-llgreen}{rwth-green!25}
\definecolor{rwth-cyan}{cmyk}{1,0,.4,0}\colorlet{rwth-lcyan}{rwth-cyan!50}\colorlet{rwth-llcyan}{rwth-cyan!25}
\definecolor{rwth-teal}{cmyk}{1,.3,.5,.3}\colorlet{rwth-lteal}{rwth-teal!50}\colorlet{rwth-llteal}{rwth-teal!25}
\definecolor{rwth-gold}{cmyk}{.35,.46,.7,.35}
\definecolor{rwth-silver}{cmyk}{.39,.31,.32,.14}\colorlet{rwth-lsilver}{rwth-silver!50}

\usepackage{hyperref}
\usepackage[capitalise,noabbrev]{cleveref}

\usepackage{graphicx}
\usepackage{color}

\begin{document}
\renewcommand{\P}[1]{\ensuremath{\mathcal{P}(#1)}}
\newcommand{\B}[1]{\ensuremath{\mathcal{B}(#1)}}
\newcommand{\A}{\ensuremath{A}}
\renewcommand{\L}{\ensuremath{\mathcal{L}}}
\renewcommand{\S}{\ensuremath{\mathcal{S}}}
\renewcommand{\SS}{\ensuremath{\mathrm{S}}}
\newcommand{\U}{\ensuremath{\mathcal{U}}}
\newcommand{\dcup}{\ensuremath{\overset{.}{\cup}}}
\renewcommand{\bot}{\scriptsize\ensuremath{\StopWatchEnd}}
\newcommand{\sat}{\ensuremath{\mathrm{sat}}}

\newcommand{\tree}[1]{\ensuremath{#1}}
\newcommand{\labeled}[1]{\ensuremath{\ell(#1)}}
\newcommand{\universe}[1][]{{\ensuremath{\mathcal{A}_{#1}}}}
\newcommand{\nodes}[1]{{\ensuremath{\mathcal{S}_{#1}}}}
\newcommand{\cost}{\ensuremath{\bar{\kappa}}}
\newcommand{\taus}{{\ensuremath{\mathrm{taus}}}}
\newcommand{\cshortest}[1]{\ensuremath{\kappa(#1)}}
\newcommand{\expl}[1]{\ensuremath{\mathrm{explain}(#1)}}
\newcommand{\lca}[2]{\ensuremath{\mathrm{lca}(#1,#2)}}
\newcommand{\leafs}{\ensuremath{\mathrm{leafs}}}
\newcommand{\base}{\ensuremath{S}}
\newcommand{\C}{\ensuremath{\mathcal{C}}}
\newcommand{\N}{\ensuremath{\Sagn}}
\newcommand{\Tex}{\ensuremath{\tree{T}^{ex}}}
\newcommand{\suc}{\ensuremath{\mathrm{succ}}}
\newcommand{\all}{\ensuremath{\mathrm{all}}}
\newcommand{\blocks}{\ensuremath{\mathrm{blocks}}}
\newcommand{\trans}{\ensuremath{\mathrm{trans}}}
\newcommand{\ltrans}{\ensuremath{\mathrm{ltrans}}}
\newcommand{\strans}{\ensuremath{\mathrm{strans}}}
\newcommand{\execs}{\ensuremath{\mathrm{execs}}}
\newcommand{\highi}{\ensuremath{\mathrm{high_i}}}
\newcommand{\highb}{\ensuremath{\mathrm{high_b}}}
\newcommand{\agn}{\ensuremath{\gamma}}
\newcommand{\sagn}{\ensuremath{\delta}}
\newcommand{\Sagn}{\ensuremath{\Delta}}
\newcommand{\alg}{\texttt{PTSA}}
\newcommand{\pna}{\texttt{PNA}}
\newcommand{\pn}{\texttt{PN}}

\newcommand{\review}[2]{\replaced{#2}{#1}}

\let\oldupharpoonright\upharpoonright
\renewcommand{\upharpoonright}{\ensuremath{\downharpoonright}}

\newcolumntype{P}[1]{>{\centering\arraybackslash}p{#1}}

\renewcommand{\question}[1]{{\textcolor{red}{\textbf{#1}}}}

\newenvironment{proofs}{%
  \renewcommand{\proofname}{Proof Sketch}\proof}{\endproof}

\newcommand{\any}{\tikz[baseline=-0.5ex]{\draw[] (0,0) -- (0.4,0);}}
\newcommand{\sdots}{\tikz[baseline=-0.5ex]{\draw[decorate,decoration={zigzag,segment length=0.97mm,amplitude=0.5mm}] (0,0) -- (0.4,0);}}
\newcommand{\anyy}{\tikz[baseline=-0.5ex]{\draw[] (0,0) -- (0.3,0);}}
\newcommand{\bars}{\tikz[baseline=-0.5ex]{\draw[draw opacity=0] (0,0) -- (0.,0);}}
\newcommand{\lldots}{\tikz[baseline=0.3ex]{\node[] at (0,0) {$\ldots$};}}

\newcommand{\changea}[3]{\makecell{\tikz[baseline=-0.5ex,text width=0.4cm, anchor=west,align=right]{\node[draw=none] (p1) at (0,0) {\ensuremath{#1}}; \node[draw=none] (p2) at (0,-0.3) {\ensuremath{#2}}; \draw[draw opacity=0] ($(p2.east)+(-0.2,0)$) edge[<-,bend right=30] node[midway,right,text width=0.4cm]{\tiny\ensuremath{#3}} ($(p1.east)+(-0.2,0)$);}}}
\newcommand{\changeb}[3]{\makecell{\tikz[baseline=-0.5ex,text width=0.8cm, anchor=west,align=right]{\node[draw=none] (p1) at (0,0) {\ensuremath{#1}}; \node[draw=none] (p2) at (0,-0.3) {\ensuremath{#2}}; \draw[draw opacity=0] ($(p2.east)+(-0.2,0)$) edge[<-,bend right=30] node[midway,right,text width=0.4cm,align=left]{\tiny\ensuremath{#3}} ($(p1.east)+(-0.2,0)$);}}}
\newcommand{\changec}[3]{\makecell{\tikz[baseline=-0.5ex,text width=1.1cm, anchor=west,align=right]{\node[draw=none] (p1) at (0,0) {\ensuremath{#1}}; \node[draw=none] (p2) at (0,-0.3) {\ensuremath{#2}}; \draw[draw opacity=0] ($(p2.east)+(-0.2,0)$) edge[<-,bend right=30] node[midway,right,text width=0.4cm,align=left]{\tiny\ensuremath{#3}} ($(p1.east)+(-0.2,0)$);}}}

\newcommand{\lccref}[1]{\lcnamecref{#1}~\labelcref{#1}}

\renewcommand{\vec}[2]{\ensuremath{\big(\begin{smallmatrix}
  #1 \\
  #2
\end{smallmatrix}\big)}}

\newcommand{\svec}[2]{\ensuremath{(\begin{smallmatrix}
  #1 \\
  #2
\end{smallmatrix})}}

\Crefname{lemma}{Lemma\strut}{Lemma\strut}

\newcommand{\ExternalLink}{%
    \tikz[x=1.2ex, y=1.2ex, baseline=-0.05ex]{%
        \begin{scope}[x=1ex, y=1ex]
            \clip (-0.1,-0.1) 
                --++ (-0, 1.2) 
                --++ (0.6, 0) 
                --++ (0, -0.6) 
                --++ (0.6, 0) 
                --++ (0, -1);
            \path[draw, 
                line width = 0.5, 
                rounded corners=0.5] 
                (0,0) rectangle (1,1);
        \end{scope}
        \path[draw, line width = 0.5] (0.5, 0.5) 
            -- (1, 1);
        \path[draw, line width = 0.5] (0.6, 1) 
            -- (1, 1) -- (1, 0.6);
        }
    }
\title{Technical Report with Proofs for\\A Full Picture in Conformance Checking: Efficiently Summarizing All Optimal Alignments}
\titlerunning{Proofs: Efficiently Summarizing All Optimal Alignments}
%
\author{Philipp Bär\inst{1}\orcidID{0009-0001-5277-0766} \and
Moe T. Wynn\inst{2}\orcidID{0000-0001-7743-5772}\and
Sander J. J. Leemans\inst{1,3}\orcidID{0000-0002-5201-7125}} 
%
\authorrunning{P. Bär et al.}
%
\institute{RWTH Aachen University, Aachen, Germany\\
\email{philipp.baer@rwth-aachen.de, s.leemans@bpm.rwth-aachen.de} \and
Queensland University of Technology, Brisbane, Australia\\
\email{m.wynn@qut.edu.au} \and
Fraunhofer FIT, Sankt Augustin, Germany}
\maketitle              
%
%
\setcounter{section}{3}
\section{Summarizing Alignments with Skip Alignments}
\setcounter{subsection}{0}
\subsection{Skip Alignments}
\setcounter{lemma}{6}
\begin{lemma}[Finite]\label{lem:sfinite}
\!\!A~trace~and~an~SB\!WF\!-net~have~finitely~many~skip~\mbox{alignments}.
\end{lemma}
\begin{proof}
    Assume a trace $\sigma$ has infinitely many skip alignments on an SBWF-net $N$. 
    The number of synchronous and log moves is limited by $|\sigma|$. Since $N$ is finite, only finitely many different skip moves exist. Hence, for every $n > 0$ there must exist a skip alignment $\sagn$ such that $|\sagn| > n$. Only loop blocks in $N$ can contribute to an unbounded length of $\sagn$. Therefore, from some $m > 0$ on, all skip alignments $\sagn$ with $|\sagn| > n$ and $n \geq m$ execute at least one loop block $N' \in N$ by consecutively skipping over both its children, contradicting the definition of $\S(N')$.
\end{proof}

\subsection{Transforming Alignments to Skip Alignments}
\setcounter{lemma}{8}
\begin{lemma}[Transformation to Skip Alignments]\label{lem:transform}
    Let $\sigma$ be a trace and $N$ an SBWF-net. Let $\agn$ be an alignment for $\sigma$ on $N$ and let $\agn'$ be the mix alignment obtained by replacing all model moves in $\agn$ with skip moves. Then, the exhaustive application of transformation rules to $\agn'$ results in a skip alignment~$\sagn$~for~$\sigma$~on~$N$.
\end{lemma}
\begin{proof}
    Every move in $\sagn$ is a log, synchronous, or skip move. It holds that $\pi_1(\sagn)_{\upharpoonright_{\neq \gg}} = \pi_1(\agn)_{\upharpoonright_{\neq \gg}}=\sigma$ since only skip moves are modified in the transformation. Furthermore, because $\agn$ is an alignment on $N$, $\pi_2(\agn)_{\upharpoonright_{\neq \gg}}$ and $\pi_2(\agn')_{\upharpoonright_{\neq \gg}}$ are executions of $N$. Cycling rules trivially maintain that property. Lifting rules replace moves in $\agn'$ only at the positions their children where executed at, i.e., the application of any rule ensures that $\pi_2(\sagn)_{\upharpoonright_{\neq \gg}}$ is an execution of $N$. In $\sagn$, every skip move is maximally general, otherwise one of (L1-L4) could be applied and $\sagn$ contains no cycles, otherwise (C1-C2) could be applied. Hence, $\pi_2(\sagn)_{\upharpoonright_{\neq \gg}} \in \S(N)$ and $\sagn$ is a skip alignment for $\sigma$ on $N$.
\end{proof}

\setcounter{subsection}{3}
\subsection{A Canonical Normal Form for Skip Alignments}
\setcounter{proposition}{0}
We show that reduction rules may be applied in any order resulting in the same unique skip alignment (\emph{canonicity}). To that end, we first prove \emph{termination} and then \emph{local confluency}, i.e., when applying two rules yields two different skip alignments, then there are subsequent reduction rules to make them equal again.
\begin{proposition}[Termination of Reduction Rules]\label{lem:termination}
    Let $\sagn$ be a skip alignment. Repeated application of the reduction rules to $\sagn$ is terminating.
\end{proposition}
\begin{proof}
    We give a weighting function that monotonically decreases with every application of a rule.
    With every application of (R1) or (R2), the size of the set of pairs $t_1 = \{ (i,j) \mid 1 \leq i < j \leq |\sagn| \land \sagn_i = \vec{\gg}{a} \land \sagn_j \in \{ \vec{e}{\gg}, \vec{e}{a} \} \}$ decreases. None of (R1-R3) increases the size of this set again.
    With every application of (R3), the size of the set of pairs $t_2 = \{ (i,j) \mid 1 \leq i < j \leq |\sagn| \land \sagn_i = \vec{\gg}{s(N_1)} \land \sagn_j = \vec{\gg}{s(N_2)} \land N_2 < N_1 \}$ decreases. Only the application of (R2) may increase $t_2$ by at most $|\sagn|$ elements.
    Hence, $(|\sagn|\!+\!1)\cdot |t_1| + |t_2|$ is monotonically decreasing with every rule application.
\end{proof}
We prove local confluency for every pair of rules where the left sides overlap. We show that after applying one of the rules the resulting skip alignment can be made equal to the one after applying the other rule.
\begin{proposition}[Local Confluency of Reduction Rules]\label{lem:confluency}
    The reduction rules are locally confluent.
\end{proposition}
\begin{proof}
    We only inspect moves where there can be overlapping rules, i.e., (R2,R3) and (R2,R2).
    Let $\sagn = \langle \any{} \vec{\gg}{s(N_1)} \sdots\allowbreak \vec{\gg}{s(N_2)}, \vec{\gg}{s(N_3)} \sdots \vec{e}{a} \any{} \rangle$. We show that independent of the order in which we apply (R2,R3) to $\langle \vec{\gg}{s(N_2)},\vec{\gg}{s(N_3)} \rangle$~or~(R2,R2) to $\langle \vec{\gg}{s(N_1)} \sdots \vec{\gg}{s(N_2)} \rangle$, the resulting skip alignments can be made equal again.
    \begin{itemize}[nosep]
        \item First applying (R2) pushing $\vec{\gg}{s(N_3)}$ over $\vec{e}{a}$ is equal to swapping $\vec{\gg}{s(N_2)}$ with $\vec{\gg}{s(N_3)}$ by (R3) and shifting back $\vec{\gg}{s(N_3)}$ afterwards with (R2).
        \item First applying (R2) to $\vec{\gg}{s(N_2)}$ pushing the move after $\vec{e}{a}$ and then applying (R2) to $\vec{\gg}{s(N_1)}$ is equal to applying (R2) first to $\vec{\gg}{s(N_1)}$ and then to $\vec{\gg}{s(N_2)}$, followed by an additional application of (R3) to the order $\langle \vec{e}{a},\vec{\gg}{s(N_1)},\allowbreak\vec{\gg}{s(N_2)} \rangle$ of the first (R2,R2) if the constraint $[N_2 < N_1]$ holds,~otherwise we apply (R3) to the result $\langle \vec{e}{a},\vec{\gg}{s(N_2)},\vec{\gg}{s(N_1)} \rangle$ of the second (R2,R2).
    \end{itemize}
\end{proof}
Canonicity follows from both propositions together with Newman's Lemma \cite{newman}.

\setcounter{lemma}{12}
\begin{lemma}[Unique Coinciding Normal Form]\label{lem:unique}
    Every alignment coincides with a single skip alignment in normal form.
\end{lemma}
\begin{proof}
    We show that for an alignment $\agn$ the set of corresponding skip alignments $\C(\agn)$ coincides with a single normal form. Let $\sagn$ and $\hat{\sagn}$ in $\C(\agn)$ be corresponding skip alignments of $\agn$. With canonicity, they reduce to coinciding normal forms $\sagn'$ and $\hat{\sagn}'$. Assume $\sagn' \neq \hat{\sagn}'$. Both normal forms share the same moves, so their order must be different.
    At the position $i$ of their first difference, one of $\sagn'$ and $\hat{\sagn}'$ has to perform a skip move. W.l.o.g., we assume $\sagn'_i = \vec{\gg}{s(N_1)}$.
    \begin{itemize}[nosep]
        \item If $\hat{\sagn}'_i = \vec{e}{\gg}$ performs a log move, then this move appears at a position $\sagn'_j$ with $j>i$. Because no synchronous move can take place in between positions~$i$ and $j$ in $\sagn'$, (R1) can be applied contradicting the normal form.
        \item If $\hat{\sagn}'_i = \vec{e}{a}$ performs a synchronous move, then the skip move $\sagn'_i$ is not required before performing the synchronous move in $\hat{\sagn}'$, so (R2) can be applied to $\sagn'$ contradicting the normal form.
        \item If $\hat{\sagn}'_i = \vec{\gg}{s(N_2)}$ performs a skip move, then this move is performed at a position $j>i$ in $\sagn'$.
        If there is a synchronous move in $\sagn'$ in between $i$ and $j$, then $\vec{\gg}{s(N_2)}$ is independent from the synchronous move, i.e., (R2) can be applied to $\hat{\sagn}'$ contradicting the formal form.
        If there is a log move in $\sagn'$ in between $i$ and $j$, then (R1) can be applied to $\sagn'$ contradicting to the normal form.
        Otherwise, there are only skip moves in between positions $i$ and $j$. Similarly, only skip moves appear in between $\vec{\gg}{s(N_2)}$ and $\vec{\gg}{s(N_1)}$ in $\hat{\sagn}'$. As both orders are valid (both are skip alignments for the same trace), $N_1$ and $N_2$ must be descendants in different branches of a parallel block. Hence, the order of $\vec{\gg}{s(N_2)}$ and $\vec{\gg}{s(N_1)}$ depends on the total order of blocks, i.e., (R3) can be applied to one of $\sagn'$ and $\hat{\sagn}'$ contradicting the normal form.
    \end{itemize}
\end{proof}

\subsection{Computing All Optimal Skip Alignments in Normal Form}
\setcounter{lemma}{15}
\begin{lemma}[Correctness \& Completeness]
    For any trace and SBWF-net, the extended \textsc{A-star} search computes all optimal skip alignments in normal form.\label{lem:cc}%
\end{lemma}

\begin{proof}
    Let $\sigma$ be a trace and $N$ an SBWF-net. We apply the extended \textsc{A-star} search to our state space starting in the initial state $(\langle \rangle, \sigma)$. First, we show correctness, then, we show completeness of the computations.

    Let $(\sagn, \langle \rangle)$ be a cost-minimally reached goal state. Then, $\sagn$ is a skip alignment in normal form for $\sigma$ on $N$:
    \begin{itemize}
        \item The trace $\sigma$ is covered by $\sagn$, i.e., $\pi_1(\sagn)_{\upharpoonright_{\neq \gg}} = \sigma$, since Algorithm 1 iteratively consumes $\sigma$ to result in $\sigma'=\langle \rangle$. Furthermore, $\pi_2(\sagn)_{\upharpoonright_{\neq \gg}} \in \S(N)$ because for every non-log move appended in the construction of $\sagn$, we check extendability of the current model execution in the skip language of $N$ (lines~2 and~6). States explored in the search consist only of skip moves (T1,T2), synchronous moves (T2), and log moves (T3), i.e., $\sagn$ is a skip alignment for $\sigma$ on $N$.

        \item It additionally is a skip alignment in normal form: (T3) cannot break the normal form, because the log move is appended to some partial skip alignment, which by construction of Algorithm 1 does not end with a skip move. (T1) and (T2) explicitly ensure by calling the Boolean test function $\mathrm{r}$, that the resulting (partial) skip alignment can be extended to a skip alignment in normal form. This also holds for $\sagn$ which can only be extended by the empty sequence $\langle \rangle$, i.e., $\sagn$ is in normal form.

        \item Optimality of $\sagn$ follows from the guarantee that \textsc{A-star} only returns optimal solutions and the fact that edge costs between nodes in the search space match skip alignment move costs.
    \end{itemize}
    Because $\sagn$ is a skip alignment, it is in normal form, and it is optimal, the extended \textsc{A-star} search on our search space is correct.

    Assume $\sagn$ is an optimal skip alignment in normal form for $\sigma$ on $N$ and $\sagn$ is not found by the extended \textsc{A-star} search on our search space. Then, there exists a largest prefix $\sagn'$ of $\sagn$ and a state $(\sagn', \sigma')$ explored in the search, such that there is no cost-minimal path from $(\sagn', \sigma')$ to the goal state $(\sagn, \langle \rangle)$. Because edge costs are move costs and because the extended \textsc{A-star} search algorithm is complete, there exists \emph{no path at all} from $(\sagn', \sigma')$ to $(\sagn, \langle \rangle)$. Since $\sagn'$ is the maximal prefix of $\sagn$ with that property, $(\sagn', \sigma')$ is \emph{missing a successor}, i.e., Algorithm 1 is incomplete. We inspect the missing suffix $\sagn''$ of $\sagn$, i.e., $\sagn = \sagn' \cdot \sagn''$, and distinguish three cases:
    \begin{itemize}
        \item If $\sagn''$ consists only of skip moves, then this $\sagn''$ is discovered in Algorithm 1 as all skip paths $p \in \SS(N)^*$ are considered (line~2), and a suitable successor is computed by Algorithm 1.
        \item If $\sagn''$ starts with a log move, then this move is covered (line~9) and a suitable successor is computed by Algorithm 1.
        \item Otherwise, $\sagn''$ starts with an unbounded number of skip moves followed by a synchronous move. Because in Algorithm 1, all matching transitions $t \in T$ (line~5) and all skip paths $p \in \SS(N)^*$ (line~6) are inspected, this move is covered and a suitable successor is computed by Algorithm 1.
    \end{itemize}
    Note that $\sagn''$ cannot start with another sequence of moves, because $\sagn$ is in normal form. This contradicts the assumption that the set of successors of $(\sagn', \sigma')$ computed by Algorithm 1 is incomplete.

    Hence, no such undiscovered $\sagn$ exists and the extended \textsc{A-star} search on our search space is complete.
\end{proof}
\begin{credits}
\subsubsection{\ackname} \!Part of this research \review{is}{was} funded by the Hans~Hermann~\mbox{Voss-Stiftung.}
\end{credits}

\bibliographystyle{splncs04}
\bibliography{bib}

\begin{thebibliography}{1}
\providecommand{\url}[1]{\texttt{#1}}
\providecommand{\urlprefix}{URL }
\providecommand{\doi}[1]{https://doi.org/#1}

\bibitem{newman}
Newman, M.H.A.: On theories with a combinatorial definition of "equivalence".
  Annals of Mathematics  (1942)

\end{thebibliography}

\end{document}